\documentclass[twoside]{article}
\usepackage[square,sort,comma,numbers]{natbib}
\usepackage{amssymb,amstext,amsthm}
\usepackage{wrapfig}
\usepackage{amssymb} 
\usepackage{amsmath} 
\usepackage[utf8]{inputenc} 
\usepackage{graphpap}
\usepackage{geometry}
\usepackage{graphics}
\usepackage{graphicx}
\usepackage{subfigure} 
\usepackage{blkarray} 
\usepackage{bm} 

\newtheorem{theorem}{Theorem}
\newtheorem{definition}[theorem]{Definition}

\newcommand{\R}{\mathbb R}

\title{\sc Stochastic safety radius on Neighbor-Joining method and
  Balanced Minimal Evolution on small trees}
\author{{\small \bf Jing Xi}\\
\small Department of Mathematics\\ \small North Carolina State University\\ \small jxi2@ncsu.edu
\and
 {\small \bf Jin Xie} \\ \small Department of Statistics \\ \small
 University of Kentucky\\
\small jin.xie@uky.edu
\and
 {\small \bf Ruriko Yoshida} \\ \small Department of Statistics \\ \small
 University of Kentucky\\
\small ruriko.yoshida@uky.edu
\and
 {\small \bf Stefan Forcey} \\ \small Department of Mathematics \\ \small
 University of Akron\\
\small sforcey@uakron.edu
}
\date{}
\begin{document}
\maketitle

\begin{abstract}
A distance-based method to reconstruct a phylogenetic tree with $n$ leaves
takes a distance matrix, $n \times n$ symmetric matrix with $0$s in the
diagonal, as its input and reconstructs a tree with $n$ leaves using tools in
combinatorics. A safety radius is a radius from a tree metric (a distance
matrix realizing a true tree) within which the input distance matrices must all
lie in order to satisfy a precise combinatorial condition under which the
distance-based method is guaranteed to return a correct tree.  A stochastic
safety radius is a safety radius under which the distance-based method is
guaranteed to return a correct tree within a certain probability. In this paper
we investigated stochastic safety radii for the neighbor-joining (NJ) method
and balanced minimal evolution (BME) method for $n =  5$.
\end{abstract}

\section{Introduction}
A \emph{phylogenetic tree} (or \emph{ phylogeny}) on the set $X = [n]$ is a
graph which summarizes the relations of evolutionary descent between different
species, organisms, or genes. Phylogenetic trees are useful tools for
organizing many types of biological information, and for reasoning about events
which may have occurred in the evolutionary history of an organism.  There has
been much research on phylogenetic tree reconstructions from alignments, and
{\em distance-based} methods are some of the best-known phylogenetic tree
reconstruction methods.

Once we compute pairwise distances $\forall (x, y) \in X\times X$
from an alignment,
we can reconstruct a phylogenetic tree
via distance-based methods.
 In contrast with parsimony methods,
 distance-based methods have been shown to be statistically
consistent in all settings  (such as the long branch attraction)
\cite{Felsenstein1978,consis,Gascuel2003,Gascuel2009}. Distance-based methods
also have a huge speed advantage over parsimony and likelihood methods in terms
of computational time, and hence enable the reconstruction of trees with large
numbers of taxa.  However, a distance-based method is not a perfect method to
reconstruct a phylogenetic tree from the input sequence data set: in the
process of computing a pairwise distance, we ignore interior nodes of a tree as
well as a tree topology, and thus we lose information from the input sequence
data sets. Therefore it is important to understand how a distance based method
works and how robust it is with noisy data sets.

One way to measure its robustness is called the safety radius.  A safety radius
is a radius from a tree metric (a distance matrix realizing a true tree) within
which the input distance matrices must all lie in order to satisfy a precise
combinatorial condition under which the distance-based method is guaranteed to
return a correct tree.  More precisely, we have the following definition.
\begin{definition}\label{safety1}
Suppose we have a vector representation of all pairwise distances $\delta \in \R^{n
  \choose 2}$ and suppose $d_{T, w}:=(d_{xy})_{x, y \in X}$, where $T \in \tau_n$, $\tau_n$
is the set of all phylogenetic unrooted trees with leaves $X = [n]$,
and $w \in \R^{2n-3}_+$, where $\R_+$ is the set of all non-negative
real numbers,  is a vector representation of the set of branch
lengths in $T$, is a tree metric, i.e., $d_{xy} \geq 0$ is the total of branch
lengths in the unique path from a leaf $x$ to a leaf $y$ in $T$.  Let $w_{\min}$ be the smallest interior branch length in $T$.
Then a method $M$ for reconstructing a phylogenetic $X$-tree from each
distance matrix $\delta$ on $X$ is said to have a $l_{\infty}$ safety radius $\rho_n$ if for any binary phylogenetic tree $T$ with $n$ leaves we have:
\[
||\delta - d_{T, w}||_\infty < \rho_n \cdot w_{\min} \Rightarrow
  M(\delta) = T.
\]
\end{definition}
Notice that the definition of the safety radius defined in Definition
\ref{safety1} is deterministic even though the input data $\delta$ is
a multivariate random variable.  Thus, this is more meaningful to
define in terms of probability distribution.  Thus, in 2014
Steel and Gascuel introduced a notion of {\em stochastic safety
  radius} \cite{steel2014}.

\begin{definition}[Stochastic safety radius]\label{safety2}
Suppose we allow $\sigma^2$ to depend on $n$:
$
\displaystyle\sigma^2 = \frac{c^2}{\log(n)},
$
for some value $c\neq 0$.
For any $ \eta> 0$, we say that a
distance-based tree reconstruction method $M$ has $\eta$-stochastic safety radius
$s = s_n$ if for every binary phylogenetic $X$-tree $T$ on n leaves, with minimum
interior edge length $w_{\min}$, and with the distance matrix $\delta$ on $X$ described by the
random errors model, we have
\[
c < s\cdot w_{\min} \Rightarrow P(M(\delta) = T) \geq 1 - \eta.
\]
\end{definition}

In this paper we focus on two distance-based methods, namely {\em
  neighbor-joining} (NJ) method and {\em balanced minimal evolution} (BME)
method.
In 2002, Desper and Gascuel introduced a BME
principle, based on a branch length estimation scheme of Pauplin
\cite{Pauplin}.  The guiding principle of minimum evolution tree reconstruction
methods is to return a tree whose total length (sum of branch lengths) is
minimal, given an  input dissimilarity map.  The BME method is a special case
of these  distance-based methods  wherein branch lengths are estimated  by a
weighted least-squares method (in terms of the input $\delta$ and the
tree $T \in \tau_n$ in question)  that puts more emphasis on shorter
distances than longer
ones. Each labeled tree topology gives rise to a vector, called herein {\em the
BME vector,}  which is obtained from Pauplin's formula.
In 2000, Pauplin
showed that the BME method is equivalent to optimizing a linear function, the
dissimilarity map,  over the BME representations of binary trees, given by the
BME vectors \cite{Pauplin}. Eickmeyer et. al. defined the {\it $n^{th}$ BME
polytope} as the convex hull of the BME vectors  for all binary trees on a
fixed number $n$ of taxa.  Hence the BME method is equivalent to optimizing a
linear function, namely, the input distance matrix $\delta$, over a BME
polytope.  They characterized  the behavior of the BME
phylogenetics on such data sets using the BME polytopes and the {\it BME
cones}, i.e., the normal cones of the BME polytope.

The study of related geometric structures, the BME cones, further clarifies
the nature of the link between phylogenetic tree reconstruction using the BME
criterion and using the NJ Algorithm developed by Saitou and Nei
\cite{Saitou1987}.  In 2006, Gascuel and
Steel showed that the NJ Algorithm, one of the most popular phylogenetic tree
reconstruction algorithms, is a greedy algorithm for finding the BME tree
associated to a  distance matrix $\delta$ \citep{Steel2006}.  The NJ
Algorithm relies on a particular criterion for iteratively selecting cherries;
details on cherry-picking and the NJ Algorithm are recalled later in the
paper.  In 2008, based on the fact that the selection criterion for
cherry-picking is linear in the  distance matrix $\delta$ \cite{Bryant2005}, Eickmeyer
et. al. showed that the NJ Algorithm will pick cherries to merge in a
particular order and output a particular tree topology $T$ if and only if the
pairwise distances satisfy a system of linear inequalities, whose solution set
forms a polyhedral cone in  ${\R}^{n \choose 2}$ \cite{kord2009}.  They
defined such a cone as an {\em NJ cone}.  In general, the sequence of cherries
chosen by the NJ Algorithm is not unique, hence multiple  distance matrix $\delta$
will be assigned by the NJ Algorithm to a single fixed tree topology $T.$  The
set of all  distance matrix $\delta$ for which  the NJ Algorithm returns a fixed tree
topology $T$ is a union of NJ cones, however this union is not convex in
general.  Eickmeyer et. al. characterized those dissimilarity
maps for which the NJ Algorithm  returns the BME tree, by comparing the NJ
cones with the BME cones, for eight or fewer taxa \cite{kord2009}.

In this paper we use the BME cones and NJ cones in order to
investigate their stochastic safety radius for $n = 5$.
Here we assume that the multivariate random variable $\delta$ is
defined as follows:
\[
\delta_{xy} = d_{xy} + \epsilon_{xy},
\]
where $\epsilon_{xy} \sim N(0, \sigma^2)$, the Gaussian distribution with
mean $0$ and a standard deviation $\sigma >0$, are independent for all
pairwise distance
$(x, y) \in X \times X$.
This paper is
organized as follows:  Section \ref{4points} shows the probability
distribution of a random $\delta$ so that it satisfies the {\em four
  point rule} for all distinct quartets in $[n]$ with a fixed $T$.  Zarestkii in \cite{Zarestkii2965} defined
the notion of the four point rule as follows:
we select the tree topology $xy|wz$ (which means there is an internal
edge between $\{x, y\}$ and $\{w, z\}$ for a distinct $x, y, w, z
\in [n]$)  if
\begin{equation}\label{eq:4ptrule}
\delta_{xy} + \delta_{wz} < \min \{\delta_{xw} + \delta_{yz}, \,
\delta_{xz} + \delta_{yw}\}.
\end{equation}
 In Section \ref{bme} we will show multivariate
probability distribution $P(M(\delta) = T)$ where $T \in \tau_5$ is
fixed and $M$ is
the BME method, in Section \ref{nj} shows the multivariate probability distribution $P(M(\delta) = T)$ where $T \in \tau_5$ is
fixed and  $M$ is
the NJ method.  Finally in Section \ref{simulation} we will show some
computational results on these probability distributions and we have
shown the plot for the stochastic safety radii for the NJ and the BME
methods varying $\eta$ and $c$ for $n = 5$ (Figure \ref{fig:safety}). 
As shown in
Figure \ref{fig:safety} both stochastic safety radii are basically
almost identical in this case since the probability distributions
$P(M(\delta) = T)$ for the NJ and for the BME methods are almost
identically same shown in Figure \ref{fig:comparison} for $n = 5$ and
$w_{\min} = 1$.

\section{Probability distribution on ``four point
  rule''}\label{4points}

For a tree containing random errors, the pairwise distance between two leaves
is
$$\delta_{xy}=d_{xy}+\epsilon_{xy}$$ where $x$ and $y$ are different taxas of a tree, $d_{xy}$ is the true pairwise distance between
taxa $x$ and $y$, and $\epsilon_{xy}'s$ follow i.i.d. Gaussian Distribution
with mean 0 and variance $\sigma^2$. Intuitively in this section we are
computing a probability distribution such that if we select a random $\delta
\in \R^{n \choose 2}$, $\delta$ satisfies Equation \ref{eq:4ptrule} if and only
if there is an internal edge between $\{x, y\}$ and $\{w, z\}$ in $T \in
\tau_n, $ for all distinct $\{x, \, y, \, w, \, z\} \in [n]$. We find a formula
for the probability, for 5 taxa, that a tree metric with random errors still
obeys the original four-point inequalities on each subset of four leaves.

\begin{wrapfigure}{r}{80mm}
  \centering
  \includegraphics[width=1.5in]{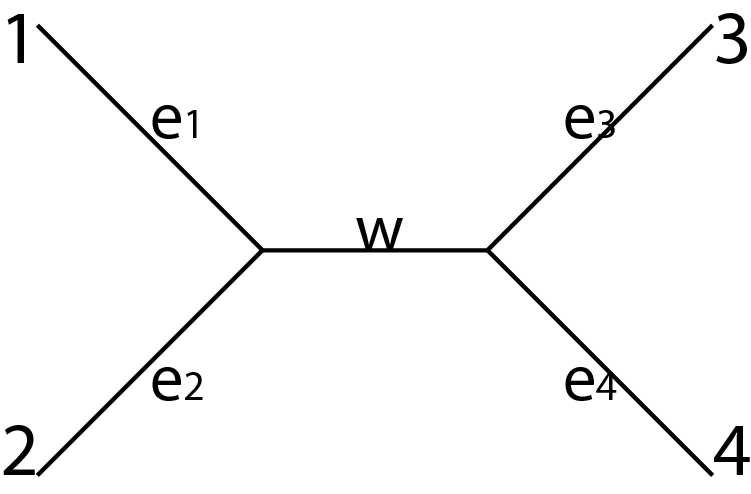}
  \caption{4 taxa tree}
  \label{fig:4taxa}
{\footnotesize }
\end{wrapfigure}
We first consider four point rule on 4 taxa tree. Suppose Figure \ref{fig:4taxa} is the true tree.
Then for a random tree, the following inequalities must be satisfied in order to return the correct tree:

\begin{equation}
\begin{array}{ccl}
\delta_{12}+\delta_{34} \le \delta_{13}+\delta_{24} \\
\delta_{12}+\delta_{34} \le \delta_{14}+\delta_{23}
\end{array}
\end{equation}

Since
\begin{equation}
\begin{array}{ccl}
\delta_{12}&=& e_1+e_2+\epsilon_{12} \\
\delta_{34}&=& e_3+e_4+\epsilon_{34} \\
\delta_{13}&=& e_1+e_3+w+\epsilon_{13} \\
\delta_{24}&=& e_2+e_4+w+\epsilon_{24} \\
\delta_{14}&=& e_1+e_4+w+\epsilon_{14} \\
\delta_{23}&=& e_2+e_3+w+\epsilon_{23}
\end{array}
\end{equation}

Then we can have
\begin{equation}\label{ineq1}
\begin{array}{ccl}
\epsilon_{12}+\epsilon_{34} \le 2w +\epsilon_{13}+\epsilon_{24} \\
\epsilon_{12}+\epsilon_{34} \le 2w +\epsilon_{14}+\epsilon_{23}
\end{array}
\end{equation}

Since $\epsilon_{xy} \stackrel{iid}{\sim} N(0,\sigma^2)$, we know
$\epsilon_{12}+\epsilon_{34}, \epsilon_{13}+\epsilon_{24},
\epsilon_{14}+\epsilon_{23} \stackrel{iid}{\sim} N(0,2\sigma^2)$. Let $f$ and
$F$ be the density and cumulative distribution functions of $N(0,1)$,
respectively. Then the probability that Inequality \ref{ineq1} is satisfied,
i.e. the probability that a random distance matrix $\delta$ returns the true
tree, equals to:

\begin{equation}
\int_{-\infty}^{\infty}f(x)[1-F(\frac{x-2w}{\sqrt{2}\sigma})]^2dx
\end{equation}

Now we consider four point rule on 5 taxa tree. Suppose the true tree is Figure \ref{fig:5taxa}. We need to check the rule on all possible combinations of four distinct leaves in this tree. It is trivial to see we only have 5 different combinations. For each of them, we could construct two inequalities similar to the way we obtained Equation \ref{ineq1}. Therefore, we have 10 inequalities for the 5 combinations of 4 distinct taxa:
\begin{equation}
\begin{array}{ccl}
\epsilon_{12}+\epsilon_{34} &\le& 2w_1 +\epsilon_{13}+\epsilon_{24} \\
\epsilon_{12}+\epsilon_{34} &\le& 2w_1 +\epsilon_{14}+\epsilon_{23} \\

\epsilon_{12}+\epsilon_{35} &\le& 2w_1 +\epsilon_{13}+\epsilon_{25} \\
\epsilon_{12}+\epsilon_{35} &\le& 2w_1 +\epsilon_{15}+\epsilon_{23} \\

\epsilon_{12}+\epsilon_{45} &\le& 2w_1+2w_2 +\epsilon_{14}+\epsilon_{25} \\
\epsilon_{12}+\epsilon_{45} &\le& 2w_1+2w_2 +\epsilon_{15}+\epsilon_{24} \\

\epsilon_{13}+\epsilon_{45} &\le& 2w_2 +\epsilon_{14}+\epsilon_{35} \\
\epsilon_{13}+\epsilon_{45} &\le& 2w_2 +\epsilon_{15}+\epsilon_{34} \\

\epsilon_{23}+\epsilon_{45} &\le& 2w_2 +\epsilon_{24}+\epsilon_{35} \\
\epsilon_{23}+\epsilon_{45} &\le& 2w_2 +\epsilon_{25}+\epsilon_{34} \\
\end{array}
\end{equation}

Let $\boldsymbol{\epsilon} \equiv (\epsilon_{12},\epsilon_{13},\epsilon_{14},\ldots,\epsilon_{45})^T_{10\times1}$ and
\[
U =  \left( \begin{array}{cccccccccc}
1 & -1 &  0 &  0 &  0 & -1 &  0 &  1 &  0 & 0  \\
1 &  0 & -1 &  0 & -1 &  0 &  0 &  1 &  0 & 0  \\
1 & -1 &  0 &  0 &  0 &  0 & -1 &  0 &  1 & 0  \\
1 &  0 &  0 & -1 & -1 &  0 &  0 &  0 &  1 & 0  \\
1 &  0 & -1 &  0 &  0 &  0 & -1 &  0 &  0 & 1  \\
1 &  0 &  0 & -1 &  0 & -1 &  0 &  0 &  0 & 1  \\
0 &  1 & -1 &  0 &  0 &  0 &  0 &  0 & -1 & 1  \\
0 &  1 &  0 & -1 &  0 &  0 &  0 & -1 &  0 & 1  \\
0 &  0 &  0 &  0 &  1 & -1 &  0 &  0 & -1 & 1  \\
0 &  0 &  0 &  0 &  1 &  0 & -1 & -1 &  0 & 1
\end{array} \right)_{10\times10},
\]
then the 10 inequalities are:
\begin{equation}\label{ineq2}
U \boldsymbol{\epsilon} \le (2w_1,2w_1,2w_1,2w_1,2w_1+2w_2,2w_1+2w_2,2w_2,2w_2,2w_2,2w_2)^T.
\end{equation}
Thus the probability that  a 5-leaved tree metric with random errors still obeys the original-four point inequalities on each subset of four leaves is the probability that inequality (\ref{ineq2}) is satisfied.

\section{Probability distribution of the output tree via the BME method}\label{bme}

This method begins with a given set of $n$ items and a symmetric (or upper
triangular) square $n\times n$ \emph{distance matrix} whose entries are
numerical dissimilarities, or distances, between pairs of items. From the
distance matrix the BME method constructs a binary tree with the $n$ items
labeling the $n$ leaves. The BME tree has the property that the distances
between its leaves most closely match the given distances between corresponding
pairs of taxa.

By ``most closely match'' in the previous paragraph we mean the following: the
reciprocals of the distances between leaves are the components of a vector
$\mathbf{c}$, and this vector minimizes the dot product $\mathbf{c}\cdot\delta$
where $\delta$ is the list of distances in the upper triangle of the distance
matrix.

More precisely: Let the set of $n$ distinct species, or taxa, be called $X.$
For convenience we will often let $X = [n] = \{1,2,\dots,n\}.$ Let vector
$\delta$ be given, having ${n \choose 2}$ real valued components $\delta_{xy}$,
one for each pair $\{x,y\}\subset X.$ There is a vector $\mathbf{c}(t)$ for
each binary tree $t$ on leaves $X,$ also having ${n \choose 2}$ components
$c_{xy}(t)$, one for each pair $\{x,y\}\subset X.$ These components are ordered
in the same way for both vectors, and we will use the lexicographic ordering:
$\delta =
(\delta_{12},\delta_{13},\dots,\delta_{1n},\delta_{23},\delta_{24},\dots,\delta_{n-1,n})
$.

We define, following Pauplin \cite{Pauplin}:
 $$\mathbf{c}_{xy}(t) = \frac{1}{2^{l(x,y)}}$$
where $l(x,y)$ is the number of internal nodes (degree 3 vertices) in the path
from leaf $x$ to leaf $y.$ The BME tree for the vector $\delta$ is the binary
tree $t$ that minimizes $\delta\cdot\mathbf{c}(t)$ for all binary trees on
leaves $X.$ Rather than the original fractional coordinates $\mathbf{c}_{xy}$
we will scale by a factor of $2^{n-2},$ giving coordinates
$$\mathbf{x}_{xy}=2^{n-2}\mathbf{c}_{xy} = 2^{n-2-l(x,y)}.$$
 Since the furthest apart any two leaves may
be is a distance of $n-2$ internal nodes, this scaling will result in integral
coordinates. Thus we can equivalently say that the BME tree for the vector
$\delta$ is the binary tree $t$ that minimizes $\delta\cdot\mathbf{x}(t)$ for
all binary trees on leaves $X.$

Consider a tree metric $d_{T}$ which arises from a binary tree $T$ with five
leaves $\{a,b,c,d,e\}.$ Let the interior edges $e_1$ and $e_2$ have lengths
$w_i=l(e_i).$

\begin{theorem}Let $T$, the tree for which $d_T$ is a tree metric, have cherries $\{a,b\}$
and $\{c,d\}.$

 Let:

 $ y_1 =   2   \epsilon_{ac} +    \epsilon_{ad}           -
\epsilon_{bc} +
\epsilon_{bd} + 3   \epsilon_{be}           -  \epsilon_{ce}\\
 y_2 =   2
\epsilon_{ac} +            \epsilon_{ae}   -  \epsilon_{bc} + 3 \epsilon_{bd}
+    \epsilon_{be}   -  \epsilon_{cd}\\
 y_3 =   -  \epsilon_{ac}
+    \epsilon_{ad} + 3   \epsilon_{ae} + 2   \epsilon_{bc} +    \epsilon_{bd}
-1  \epsilon_{ce}\\
y_4 =   -  \epsilon_{ac} + 3   \epsilon_{ad} +  \epsilon_{ae} + 2
\epsilon_{bc} +            \epsilon_{be}   - \epsilon_{cd}\\
$

Let:

$ z_1 =   -  \epsilon_{ac} + 3   \epsilon_{ad} +
\epsilon_{bd} +    \epsilon_{be}   -  \epsilon_{cd} + 2   \epsilon_{ce}\\
 z_2
= -  \epsilon_{ac} +         3   \epsilon_{ae} +            \epsilon_{bd} +
\epsilon_{be} + 2   \epsilon_{cd}   -  \epsilon_{ce}\\
z_3 =            \epsilon_{ad} +    \epsilon_{ae}   -  \epsilon_{bc} + 3
\epsilon_{bd} -  \epsilon_{cd} + 2   \epsilon_{ce}\\
z_4 =            \epsilon_{ad} +  \epsilon_{ae}   -  \epsilon_{bc} + 3
\epsilon_{be} + 2 \epsilon_{cd}   -  \epsilon_{ce}$

Then the BME method will return the correct tree $T$ if and only if:

$ 4w_2>\epsilon_{ac} +\epsilon_{bc}
+2\epsilon_{de}-\min(\epsilon_{ad}+\epsilon_{bd}+2\epsilon_{ce},
~\epsilon_{ae}+\epsilon_{be}+2\epsilon_{cd})\\
 6w_1 + 4w_2>3\epsilon_{ab}
+2\epsilon_{de}-\min(y_1,~y_2,~y_3,~y_4)\\
6w_1 +
6w_2>3\epsilon_{ab}+3\epsilon_{de}-\min(3\epsilon_{ae}+3\epsilon_{bd},~3\epsilon_{ad}+3\epsilon_{be})\\
4w_1 + 6w_2>2\epsilon_{ab} +3\epsilon_{de}-\min(z_1,~z_2,~z_3,~z_4)\\
4w_1>2\epsilon_{ab} +\epsilon_{cd}
+\epsilon_{ce}-\min(\epsilon_{ad}+\epsilon_{ae}+2\epsilon_{bc},
~2\epsilon_{ac}+\epsilon_{bd}+\epsilon_{be}) $
\end{theorem}

\begin{proof}
The BME method will return the correct tree $T$ if and only if
$$({d}_{T}+{\epsilon})\cdot\mathbf{x}(T) < ({d}_{T}+{\epsilon})\cdot\mathbf{x}(t)$$ for
all alternate trees $t.$ This is true since the 1-skeleton of the BME polytope
for $n=5$ is the complete graph on the 15 vertices.

Further, the above inequality holds iff
$${d}_T\cdot(\mathbf{x}(t)-\mathbf{x}(T))> {\epsilon}\cdot(\mathbf{x}(T)-\mathbf{x}(t))$$ for
all alternate trees $t.$ Note that all the trees with five leaves have the same
topology.

There are 14 other possible trees $t.$  These separate into 5 sets of trees for
which the left hand side of the above inequality is respectively $4w_2, 6w_1 +
4w_2, 6w_1 + 6w_2, 4w_1 + 6w_2, $ or $4w_1.$  For each of these we collect the
right hand sides, and take their maximum.
\end{proof}

\section{Probability distribution of the output tree via the
  NJ method}\label{nj}
%
%
%

\subsection{H-representation of NJ cones \cite{Rudy2008}}\label{sec:njcone}

Recall that the tree metric $d_{T,w}=(d_{xy})_{1\leq x,y \leq n}$ is a symmetric matrix with $d_{xx}=0$. We can flatten the entries in the upper triangle (diagonal entries are omitted) by columns:
\[
d_{xy}= d_{\frac{(y-1)(y-2)}{2}+x},
\]
where $1\leq x \leq n-1$ and $x+1 \leq y \leq n$. Let $\mathbf d = (d_1, d_2, \ldots, d_m)$, $m:= {n \choose 2}$, be the vector of tree metric after flattening. Notice here this flattening defines a one-to-one mapping between the indices:
\[
I_f: \{(x,y)\in\mathbb Z: 1\leq x \leq n-1, x+1 \leq y \leq n\} \rightarrow \{1,2,\ldots,m\}, \
I_f(x,y) = \frac{(y-1)(y-2)}{2}+x.
\]

In NJ algorithm, we first compute the Q-criterion (cherry picking criterion):
\[
q_{xy} = (n-2) d_{xy} - \sum\limits_{z=1}^n d_{xz} - \sum\limits_{z=1}^n d_{zy}.
\]
Similar as the flattened tree metric $\mathbf d$, the Q-criterion can also be flattened to a $m$ dimensional vector $\mathbf q$ which can be obtained from $\mathbf d$ by linear transformation:
\[
\mathbf q = A^{(n)} \mathbf d,
\]
where matrix $A^{(n)}$ is defined as:
\[
A^{(n)}_{ij} =
\begin{cases} n-4 & \text{if } i=j \\
-1 & \text{if } i\neq j \text{ and } \{x,y\}\cap\{z,w\}\neq \emptyset \\
0 & \text{else}
\end{cases},
\]
where $(x,y) = I_f^{-1}(i)$ and $(z,w)=I_f^{-1}(j)$.

Now each entry in $\mathbf q$ corresponds to a pair of nodes in $T$, the next step of NJ algorithm is to find the minimum entry of $\mathbf q$ and join the corresponding two nodes as a cherry (``cherry picking''), then these two nodes will be replaced by a new node and the tree metric will be updated as $\mathbf d'$ (the dimension is reduced). We can see NJ algorithm proceeds by picking one cherry and reducing the size of the tree metric in each iteration until a binary tree is reconstructed. Without loss of generality and for the convenience of expression, we will only give details for the first iteration and assume the cherry we pick is $(n-1,n)$ in the rest part of this section.

First, to make cherry $(n-1,n)$ be the one to be picked, $q_m = q_{I_f(n-1,n)}$ needs to be the minimum in $\mathbf q$. This means the following inequalities need to be satisfied:
\[
(I_{m-1}, -\mathbf 1_{m-1})\mathbf q \geq 0 \Longrightarrow
H^{(n)} \mathbf d \geq 0,\ H^{(n)} = (I_{m-1}, -\mathbf 1_{m-1}) A^{(n)}.
\]
Note that if an arbitrary cherry is picked, then a permutation of columns need to be assigned to $H^{(n)}$.

Then, after picking cherry $(n-1,n)$, we join these two nodes as the new node $(n-1)^*$. Again, we can produce the new reduced tree metric from $\mathbf d$ by linear transformation:
\[
\mathbf d' = R \mathbf d,
\]
where $R=(r_{ij})\in \mathbb R^{(m-n+1)\times m}$,
\[
r_{ij} =
\begin{cases} 1 & \text{if } 1\leq i = j \leq {n-2 \choose 2}\\
1/2 & \text{if } {n-2 \choose 2} +1 \leq i \leq {n-1 \choose 2},\ j=i \\
1/2 & \text{if } {n-2 \choose 2} +1 \leq i \leq {n-1 \choose 2},\ j=i+n-2 \\
-1/2 & \text{if }  {n-2 \choose 2} +1 \leq i \leq {n-1 \choose 2},\ j=m\\
0 & \text{else}
\end{cases}.
\]

At last, after including inequalities in all iterations, by the shifting lemma in \cite{Rudy2008}, we also include the following equalities: $\forall$ node $x$,
\[
s_x^T \mathbf d =0,\ (s_x)_i =
\begin{cases} 1 & \text{if } x\in I_f^{-1}(i) \\
0 & \text{else}
\end{cases}.
\]

\subsection{H-representation of 5 taxa NJ cones}

There is only one tree topology for 5 taxa tree. Therefore, without loss of generality, we assume our true tree is Figure \ref{fig:5taxa}.


\begin{figure}[ht]
  \centering
  \subfigure[The 5 taxa tree]{
    \label{fig:5taxa}     
    \includegraphics[width=1.2in]{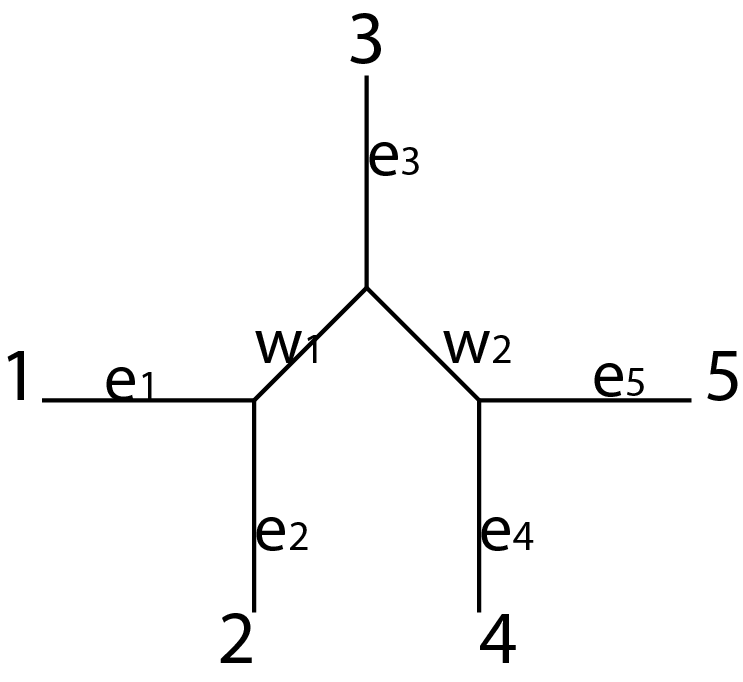}}
  \hspace{0.5in}
  \subfigure[Two orderings of picking cherries for 5 taxa tree]{
    \label{fig:5taxacp}     
    \includegraphics[width=3.8in]{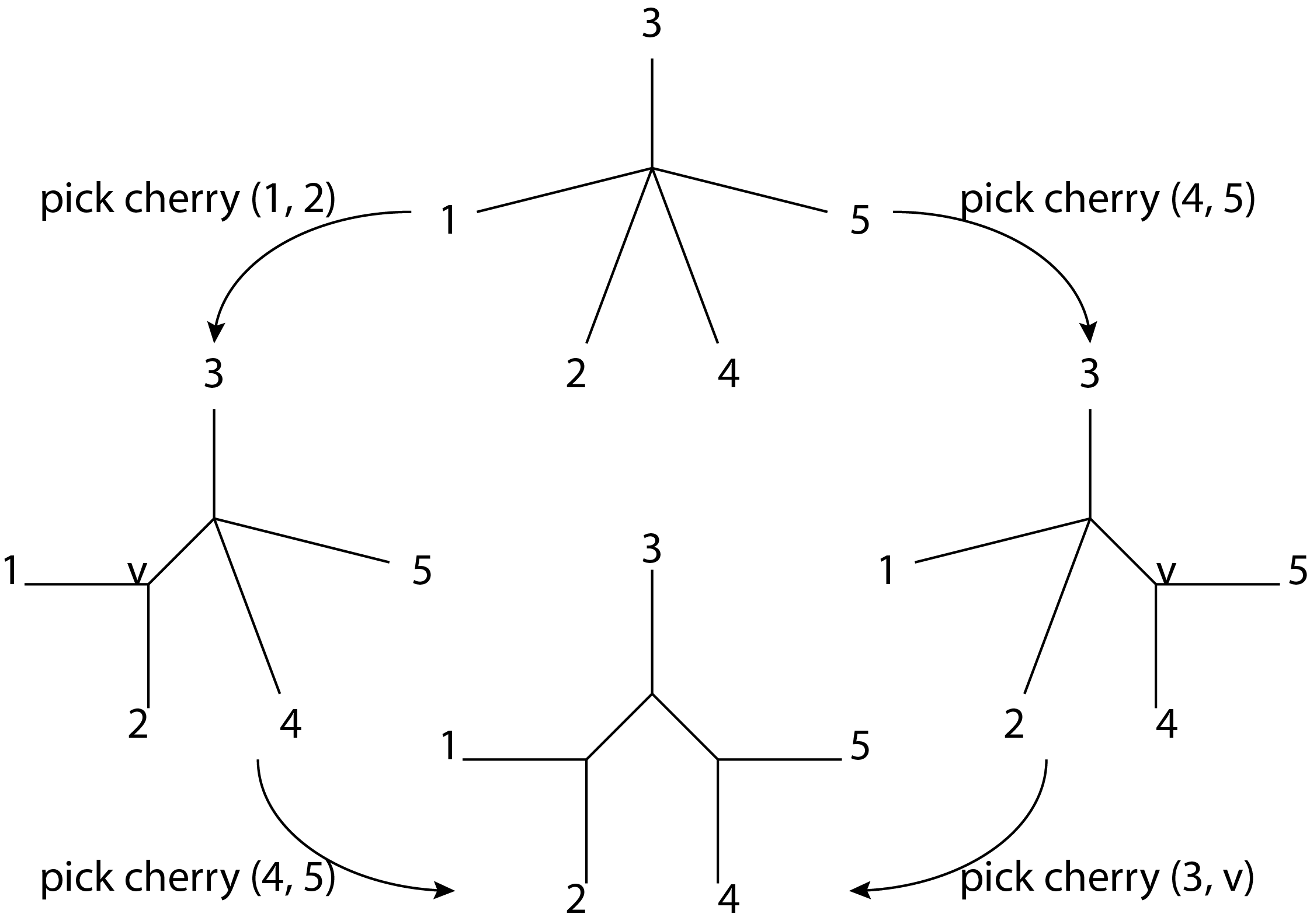}}
  \caption{The 5 taxa tree used to generate data, all edges has length 1}
  \label{fig:5taxaboth}     
{\footnotesize }
\end{figure}

For 5 taxa tree, the flattening for tree metric is:
\begin{center}
\begin{blockarray}{ccccccccccc}
xy: & 12 & 13 & 23 & 14 & 24 & 34 & 15 & 25 & 35 & 45 \\
\begin{block}{c(cccccccccc)}
$\mathbf d$ =  & $d_1$  &  $d_2$  &  $d_3$  &  $d_4$  &  $d_5$  &  $d_6$  &  $d_7$  &  $d_8$  &  $d_9$  &  $d_{10}$ \\
\end{block}
\end{blockarray}
\end{center}
In Section \ref{sec:njcone}, we can see that the permutation of columns for $H^{(n)}$ and $R$ depends on the cherry we pick. This means that we should compute NJ cones for all ordering of cherry picking (see the two orderings of cherry picking in Figure \ref{fig:5taxacp} for example).


There are four orderings of cherry picking. First we can pick cherry $(1,2)$ then pick cherry $(4,5)$, which we denote as $(1,2)-(4,5)$. Use a similar notation we have the other three: $(1,2)-((1,2), 3)$, $(4,5)-(1,2)$, and $(4,5)-(3,(4,5))$. Take the ordering $(4,5)-(3,(4,5))$ for example, use the results in Section \ref{sec:njcone} we can obtain the following linear constraints on $\mathbf d$:
\[
\left( \begin{array}{cccccccccc}
1 & -1 & -1 & 0 & 0 & 1 & 0 & 0 & 1 & -1 \\
-1 & 1 & -1 & 0 & 1 & 0 & 0 & 1 & 0 & -1 \\
-1 & -1 & 1 & 1 & 0 & 0 & 1 & 0 & 0 & -1 \\
-1 & -1 & 0 & 2 & 0 & 0 & 0 & 1 & 1 & -2 \\
-1 & 0 & -1 & 0 & 2 & 0 & 1 & 0 & 1 & -2 \\
0 & -1 & -1 & 0 & 0 & 2 & 1 & 1 & 0 & -2 \\
-1 & -1 & 0 & 0 & 1 & 1 & 2 & 0 & 0 & -2 \\
-1 & 0 & -1 & 1 & 0 & 1 & 0 & 2 & 0 & -2 \\
0 & -1 & -1 & 1 & 1 & 0 & 0 & 0 & 2 & -2 \\
-1 & 1 & 0 & 0 & 0.5 & -0.5 & 0 & 0.5 & -0.5 & 0 \\
-1 & 0 & 1 & 0.5 & 0 & -0.5 & 0.5 & 0 & -0.5 & 0
 \end{array} \right)
 \mathbf d \geq 0;
\]
\[
\left( \begin{array}{cccccccccc}
1 & 1 & 0 & 1 & 0 & 0 & 1 & 0 & 0 & 0 \\
1 & 0 & 1 & 0 & 1 & 0 & 0 & 1 & 0 & 0 \\
0 & 1 & 1 & 0 & 0 & 1 & 0 & 0 & 1 & 0 \\
0 & 0 & 0 & 1 & 1 & 1 & 0 & 0 & 0 & 1 \\
0 & 0 & 0 & 0 & 0 & 0 & 1 & 1 & 1 & 1
 \end{array} \right)
 \mathbf d = 0;
\]
Although it is not obvious to see,
 we found that the linear constraints for
 orderings $(1,2)-(4,5)$ and $(1,2)-((1,2), 3)$ are exactly the same, and the linear
 constraints for orderings $(4,5)-(1,2)$ and $(4,5)-(3,(4,5))$ are exactly the same.
  Therefore we only consider two NJ cones: the one for $(1,2)-(4,5)$ (denote as $\mathbf C_{(1,2)-(4,5)}$), and the
   one for $(4,5)-(1,2)$ (denote as $\mathbf C_{(4,5)-(1,2)}$).

\subsection{Computing the probability that NJ reconstructs the correct 5 taxa tree}

For the 5 taxa tree given in Figure \ref{fig:5taxa} under the random errors model, we know that the flattened $\delta$ should follow a multi-variate normal (MVN) distribution with mean $\mu = (2,3,3,4,4,3,4,4,3,2)$ and covariance matrix $\Sigma = \sigma^2 I_{10}$. To trace the performance of NJ algorithm under different variation, we let $\sigma^2$ to be a set of values in $(0,1]$ and then compute the probability that NJ algorithm reconstructs the correct tree for each value of $\sigma^2$.

For a given $\sigma^2$, it is trivial to see that the probability that NJ
algorithm returns the right tree is $Pr(\delta\in \mathbf C_{(1,2)-(4,5)}) +
Pr(\delta\in \mathbf C_{(4,5)-(1,2)}) - Pr(\delta\in \mathbf C_{(1,2)-(4,5)}
\cap \mathbf C_{(4,5)-(1,2)})$.

We used software \textbf{Polymake} \cite{Gawrilow2001} and verified that the dimension of $\mathbf
C_{(1,2)-(4,5)} \cap \mathbf C_{(4,5)-(1,2)}$ is lower than both of them,
therefore $Pr(\delta\in \mathbf C_{(1,2)-(4,5)} \cap \mathbf
C_{(4,5)-(1,2)})=0$. \textbf{Polymake} also gives us the facets of these two NJ
cones.
For example, the facets of $\mathbf C_{(4,5)-(1,2)}$ are:
\[
\left( \begin{array}{cccccccccc}
1 & -1 & -1 & 0 & 0 & 1 & 0 & 0 & 1 & -1 \\
-1 & -1 & 0 & 2 & 0 & 0 & 0 & 1 & 1 & -2 \\
-1 & 0 & -1 & 0 & 2 & 0 & 1 & 0 & 1 & -2 \\
0 & -1 & -1 & 0 & 0 & 2 & 1 & 1 & 0 & -2 \\
-1 & -1 & 0 & 0 & 1 & 1 & 2 & 0 & 0 & -2 \\
-1 & 0 & -1 & 1 & 0 & 1 & 0 & 2 & 0 & -2 \\
0 & -1 & -1 & 1 & 1 & 0 & 0 & 0 & 2 & -2 \\
-1 & 0 & 1 & 0.5 & 0 & -0.5 & 0.5 & 0 & -0.5 & 0 \\
-1 & 1 & 0 & 0 & 0.5 & -0.5 & 0 & 0.5 & -0.5 & 0
 \end{array} \right)
 \mathbf d \geq 0;
\]
With these facets, we can use the \textbf{R} function ``pmvnorm\{mvtnorm\}" with GenzBretz algorithm to compute $Pr(\delta\in \mathbf C_{(1,2)-(4,5)})$ and $Pr(\delta\in \mathbf C_{(4,5)-(1,2)})$.

\section{Computational experiments}\label{simulation}
In this section, we show both the theoretical and simulation probabilities that
the four point rule reconstructs the correct tree, as well as NJ algorithm and
BME method. In our computational experiments, we set all branch lengths to be $1$'s and
compute the probabilities for different values of $\sigma^2$.

\begin{wrapfigure}{r}{80mm}
  \centering
  \includegraphics[width=2.5in]{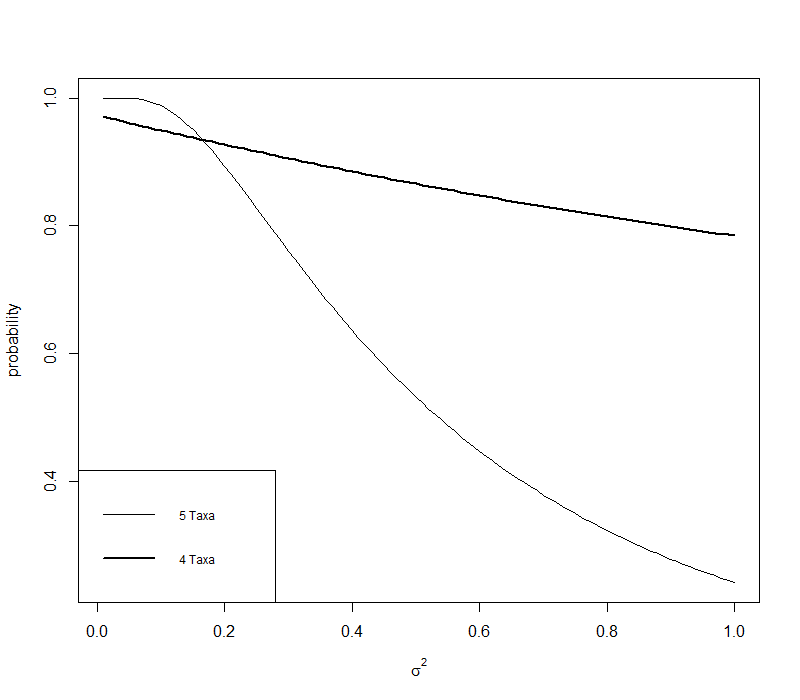}
  \caption{Theoretical probability that four point rule will return the
  correct 4 taxa tree,  and that 5 taxa tree metric with random errors still obeys
   the original four-point inequalities on each subset of four leaves.}
  \label{fig:4PR_Theory}
{\footnotesize }
\end{wrapfigure}

In Figure \ref{fig:4PR_Theory}, when $\sigma^2$ is increasing, the probability
 of 5 taxa tree will dramatically decrease faster than the probability of 4 taxa
  tree because we have more constraints to satisfy in 5 taxa tree which leads to lower probabilities.

In Figure \ref{fig:NJ_sim}, we calculated the theoretical probability
 that the NJ method reconstructs the correct 5 taxa tree based on Section \ref{nj}.
 For the simulation, we fix the 5 taxa tree in Figure \ref{fig:5taxa} with
  all branch lengths to be $1's$, and add i.i.d. normal random errors $\epsilon_{xy}'s$ to
   the pairwise distance matrix. Then we use \textbf{R} function
   ``nj\{ape\}" from the ``ape'' package in \textbf{R} \cite{APE} to reconstruct the tree.
   If the RF distance equals  0, it means that NJ method successfully returns the
    correct tree. Our simulation is based on 10,000 random trees.  Figure \ref{fig:NJ_sim} shows that
    the theoretical probabilities perfectly match the simulation result.

\begin{figure}[ht]
  \centering
  \subfigure[Theoretical Probability and Simulation for NJ on 5 Taxa Tree]{
    \label{fig:NJ_sim}     
    \includegraphics[width=2in]{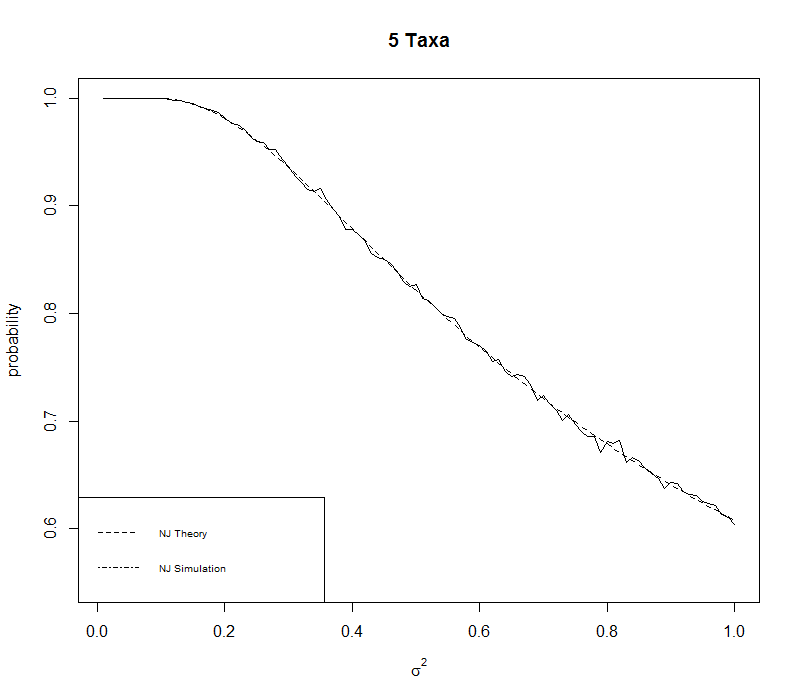}}
  \hspace{0.5in}
  \subfigure[Theoretical Probability and Simulation for BME on 5 Taxa Tree]{
    \label{fig:BMEsim}     
    \includegraphics[width=2.3in]{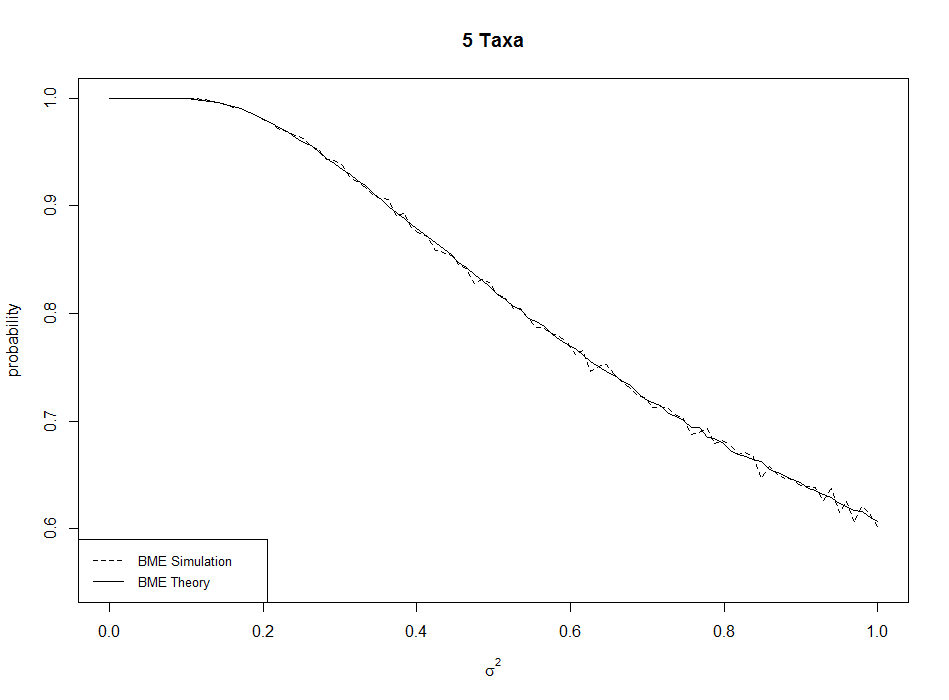}}
  \caption{Probability distributions for five leaves}
  \label{fig:sim}     
{\footnotesize }
\end{figure}


\begin{figure}[ht]
    \centering
    \includegraphics[width=4in]{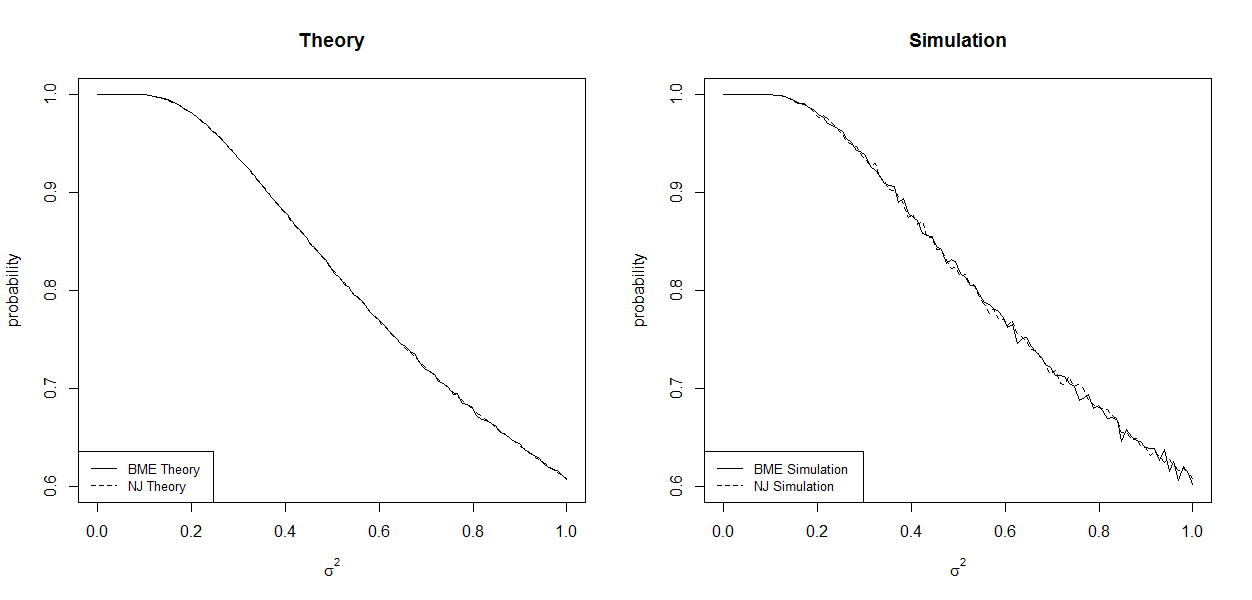}
    \caption{Comparison between BME and NJ on 5 Taxa Tree}
    \label{fig:comparison}
    {\footnotesize }
\end{figure}

\begin{wrapfigure}{r}{80mm}
    \centering
    \includegraphics[width=2.5in]{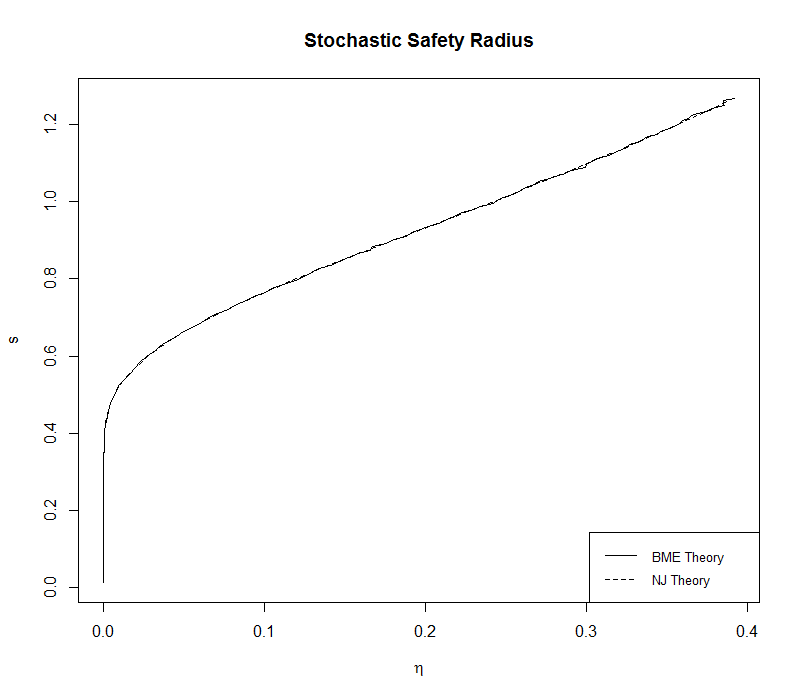}
    \caption{Stochastic safety radii for the NJ and BME methods for $n
      = 5$.  The x-axis represents $\eta$ and the y-axis represents
      the upper bound for $c/w_{\min}$ with $w_{\min} = 1$ for this experiment. }
    \label{fig:safety}
    {\footnotesize }
\end{wrapfigure}

In Figure \ref{fig:BMEsim}, we calculated the theoretical and simulated
probabilities that the BME method will return the correct 5 taxa tree. For the
theoretical probability, we generate 100,000 sets of random errors, and check
whether the theoretical rule is satisfied. In the end, we return the
percentage. For the simulation, we generate random trees in the similar way to
what we did for NJ algorithm. Then we used \textbf{R} function
``fastme.bal\{ape\}" to reconstruct the tree. Again we used RF distance to
check if the BME method successfully returned the correct tree. Our simulation
is based on 10,000 random trees. Figure \ref{fig:BMEsim} shows that the
theoretical probabilities perfectly match the simulation result.



Figure \ref{fig:comparison} shows that there is almost no difference between
BME and NJ methods on 5 taxa tree in both theoretical probabilities and
simulation result.

Figure \ref{fig:safety} shows the stochastic safety radii for the NJ
and the BME methods for $n = 5$ and $w_{\min} = 1$.  As shown in
Figure \ref{fig:safety} both stochastic safety radii are basically
almost identical in this case since the probability distributions
$P(M(\delta) = T)$ for the NJ and for the BME methods are almost
identically same shown in Figure \ref{fig:comparison}.  

\section*{Acknowledgements}
Stefan Forcey would like to thank the American Mathematical Society and the
Mathematical Sciences Program of the National Security Agency for supporting
this research through grant H98230-14-0121.\footnote{This manuscript is
submitted for publication with the understanding that the United States
Government is authorized to reproduce and distribute reprints.} 

\bibliographystyle{plain}
\bibliography{samsi}

\begin{thebibliography}{10}

\bibitem{Gascuel2009}
M.~Bordewich, O.~Gascuel, {K. T.} Huber, and V.~Moulton.
\newblock Consistency of topological moves based on the balanced minimum
  evolution principle of phylogenetic inference.
\newblock {\em IEEE/ACM Trans. Comput. Biology Bioinform.}, 6(1):110--117,
  2009.

\bibitem{Bryant2005}
D~Bryant.
\newblock On the uniqueness of the selection criterion in neighbor-joining.
\newblock {\em J. Classif.}, 22:3--15, 2005.

\bibitem{consis}
Ronald~W. DeBry.
\newblock The consistency of several phylogeny-inference methods under varying
  evolutionary rates.
\newblock {\em Mol Biol Evol}, 9(3):537--551, 1992.

\bibitem{Gascuel2003}
F.~Denis and O.~Gascuel.
\newblock On the consistency of the minimum evolution principle of phylogenetic
  inference.
\newblock {\em Discrete Applied Mathematics}, 127(1):63--77, 2003.

\bibitem{kord2009}
K.~Eickmeyer, P.~Huggins, L.~Pachter, and R.~Yoshida.
\newblock On the optimality of the neighbor-joining algorithm.
\newblock {\em Algorithms for Molecular Biology}, 3(5), 2008.

\bibitem{Rudy2008}
Kord Eickmeyer and Ruriko Yoshida.
\newblock R: Geometry of neighbor-joining algorithm for small trees.
\newblock In {\em Proceedings of the third international conference on
  Algebraic Biology}, 2008.

\bibitem{Felsenstein1978}
J.~Felsenstein.
\newblock Cases in which parsimony or compatibility methods will be positively
  misleading.
\newblock {\em Syst. Zool.}, 22:240--249, 1978.

\bibitem{Steel2006}
O~Gascuel and M~Steel.
\newblock Neighbor-joining revealed.
\newblock {\em Molecular Biology and Evolution}, 23(11):1997--2000, 2006.

\bibitem{steel2014}
O.~Gasquel and M.~Steel.
\newblock A 'stochastic safety radius' for distance-based tree reconstruction,
  2014.

\bibitem{Gawrilow2001}
E~Gawrilow and M~Joswig.
\newblock polymake: an approach to modular software design in computational
  geometry.
\newblock In {\em Proceedings of the 17th Annual Symposium on Computational
  Geometry}, pages 222--231. ACM, 2001.
\newblock June 3-5, 2001, Medford, MA.

\bibitem{Zarestkii2965}
Zarestkii K.
\newblock Reconstructing a tree from the distances between its leaves (in
  russian).
\newblock {\em Uspehi Mathematicheskikh Nauk}, 20:90--92, 1965.

\bibitem{APE}
E.~Paradis, J.~Claude, and K.~Strimmer.
\newblock A{PE}: analyses of phylogenetics and evolution in {R} language.
\newblock {\em Bioinformatics}, 20:289--290, 2004.

\bibitem{Pauplin}
Y.~Pauplin.
\newblock Direct calculation of a tree length using a distance matrix.
\newblock {\em J. Mol. Evol.}, 51:41--47, 2000.

\bibitem{Saitou1987}
N~Saitou and M~Nei.
\newblock The neighbor joining method: a new method for reconstructing
  phylogenetic trees.
\newblock {\em Molecular Biology and Evolution}, 4(4):406--425, 1987.

\end{thebibliography}





\end{document}